\documentclass{amsart}
\usepackage{sansmath}
\UseRawInputEncoding
\usepackage{amssymb,amsthm,amsmath}
\usepackage[numbers,sort&compress]{natbib}
\usepackage{color}
\usepackage{graphicx}
\usepackage{stmaryrd}
\usepackage{tikz}
\usepackage{caption}
\usepackage[latin,english]{babel}
\usepackage{braket}
\usepackage[toc,page]{appendix}

\title[ ]{Dynamical localization for  polynomial long-range hopping random operators on $\mathbb{Z}^d$ }

\author{Wenwen Jian}
\address[Wenwen Jian]{College of Arts and Sciences, Shanghai Polytechnic University, Shanghai 201209, P. R. China} \email{wwjian16@fudan.edu.cn}

\author{Yingte Sun}
\address[Yingte Sun]{School of Mathematical Sciences,
Yangzhou University,
Yangzhou 225009,
P. R. China} \email{sunyt15@fudan.edu.cn}

\newcommand{\R}{\mathbb R}

\newcommand{\Z}{\mathbb Z}

\newcommand{\C}{\mathbb C}

\newcommand{\N}{\mathbb{N}}

\theoremstyle{plain}
\newtheorem{thm}{Theorem}[section]
 
 \newtheorem{lem}[thm]{Lemma}
 
 \theoremstyle{definition}
 \newtheorem{defn}[thm]{Definition}
 \newtheorem{rem}[thm]{Remark}
 
 \numberwithin{equation}{section}
\begin{document}

\begin{abstract}
In this paper, we prove a power-law version dynamical localization for a random operator $\mathrm{H}_{\omega}$ on $\mathbb{Z}^d$ with long-range hopping. In breif, for the linear Schr\"odinger equation $$\mathrm{i}\partial_{t}u=\mathrm{H}_{\omega}u, \quad u \in \ell^2(\mathbb{Z}^d), $$  the Sobolev norm of the solution with well localized  initial state is bounded for  any $t\geq 0$.

\end{abstract}

\maketitle

\section{Introduction and the main result}
From the breaking working of Anderson \cite{And08}, a great deal of attention has paid to the  Anderson model (a linear random Schr\"odinger operator) $\mathrm{H}_{\omega}$ on  $\mathbb{Z}^d$, where $$ \mathrm{H}_{\omega}= \mathrm{H}_0+\mathrm{V}_{\omega}.$$ The operator  $\mathrm{H}_0$ is a  negative discrete  Laplacian:
$$(\mathrm{H}_0u)(n)=-\sum_{m\in\mathbb{Z}^d:\ \sum^d_{v=1}|m_v-n_v|=1}(u(m)-u(n)).$$
The potential $\mathrm{V}_{\omega}$ is a multiplication operator with a function $\mathrm{V}_{\omega}(n)$ on $\Z^d$, that $\mathrm{V}_{\omega}(n)$ are independent, identically distributed random variables. We say that the operator $\mathrm{H}_{\omega}$ has exponential localization, if its spectrum is pure point with exponential decay energy state. Namely, for some $\alpha>0$, and any energy state $\psi_{k}$, one has
\begin{equation}\label{c1}
|\psi_{k}(n)|\leq C(k)e^{-\alpha|n|},
\end{equation}
where $C(k)<+\infty$, and depend on energy state.

The mathematicians have development a bit knowledge about the localization of the random operator $ \mathrm{H}_{\omega}$.  For  $d=1$, they  proved that  the exponential localization about the random Schr\"odinger operator $\mathrm{H}_{\omega}$  for all energies. For $d \geq 2$, based on the multi-scale-analysis method(see \cite{Fro83,Kir03})  and fractional moment method(see \cite{AM93,AM15}), they can prove that the exponential localization of Anderson model at high disorder or low energies. However£¬ the physicists are more concerned with the transport properties of the model. In particular, the phenomenon is known as dynamical localization.

  Considering the random operator $\mathrm{H}_{\omega}$ on $\mathbb{Z}^d$  with pure point spectrum, the notion of dynamical localization can be reformulated as follows: for the Schr\"odinger equation on $\mathbb{Z}^d$,
\begin{equation}\label{eq1}
\mathrm{i}\partial_{t}u=\mathrm{H}_{\omega}u,  \quad  u \in \ell^2(\Z^d)
\end{equation}
with well localised initial state $u(0)$, the solution of the equation \eqref{eq1} satisfies that
\begin{equation}
\sup_{t}\|u(t)\|_{\mathrm{H}^s}=\sup_{t}(\sum_{n \in \Z^d}(1+|n|)^{2s}|u_{n}(t)|^2)^{\frac{1}{2}} < +\infty,
\end{equation}
for any $s>0$. Hence, the dynamical localization is equivalent to  the Sobolev norm of the solution is bounded for the all time.

The first rigorous proof of dynamical localization is attributed to Aizenmman\cite{AM94} by employing the fractional moment method. From the point view of multi-scale analysis technique, an effective way to obtain dynamical localization is to control the constant $C(k)$ in \eqref{c1}. In \cite{Del96}, the authors introduced the SULE condtition. Namely, all the energy state of random Schr\"odinger operator $\mathrm{H}_{\omega}$ have the form of
\begin{equation}
|\psi_{k}(n)| \leq D(\epsilon,\omega)e^{\epsilon|n_k|}e^{-\alpha|n-n_k|},
\end{equation}
where $\epsilon> 0$, $D(\epsilon,\omega)$ is an finite constant that does not depend on the energy of the state, and $n_k$ is the localization center point where $\psi_k(n)$ has its maximum. After that, the results of SULE condition was applied in \cite{Do1,Del96,J19,T01} to prove dynamical localization of some concrete models.

The above results focus on the models with exponential localization energy state. However, there are no results involving the models with power-law localization energy state. Namely,  for some $\alpha>0$, and any energy state $\psi_{k}$, one has
\begin{equation}\label{c2}
|\psi_{k}(n)|\leq C(k)|n|^{-\alpha},
\end{equation}
where $C(k)$ is finite constant depend on energy state. In \cite{Shi2020}, Shi study the $d$ dimensional random operators with long-range hopping, that is
\begin{equation}\label{h}
  \mathrm{H}_{\omega}=\lambda^{-1}\mathcal{T}+V_{\omega}(n)\delta_{nn'},\quad \lambda\geq 1,
\end{equation}
where $\lambda$ is the coupling constant describing the effect of disorder. Thorough the multi-scale analysis method, Shi proves that the energy state of random operator \eqref{h} exhibits power-law localization of energy state. In this paper, we try to extend the SULE condition to the power-law localization energy state and prove a new version of dynamical localization for the random operator \eqref{h}.

Here, we make some set-up for our main results.

$\star$ The polynomial long-range hopping operator $\mathcal{T}$ is
\begin{equation}\label{t}
  \mathcal{T} (m,n)=\left\{  \begin{aligned}
                  |m-n|^{-r}, &\ \ \ \ \text{for } m\neq n \text{ with } m,n\in\mathbb{Z}^d, \\
                   0,\ \ \ \ \ \  &\ \ \ \  \text{for } m=n\in\mathbb{Z}^d,
                   \end{aligned}
          \right.
\end{equation}
where $|n|=\max\limits_{1\leq i\leq d}|n_i|$ and $r>0$. From \cite{Geb20}, we known that the operator $\mathcal{T}$ can be viewed as a negative discrete fractional Laplacian on $\Z^d$.

 $\star$ $\{V_{\omega}(n)\}_{n\in\mathbb{Z}^d}$ is independent identically distributed (\emph{i.i.d.}) random variables (with common probability distribution $\mu$) on some probability space $(\Omega,\mathcal{F},\mathbb{P})$ ($\mathcal{F}$ a $\sigma$-algebra on $\Omega$ and $\mathbb{P}$ a probability measure on $(\Omega,\mathcal{F})$).

Throughout this paper we assume that:
\begin{itemize}
  \item We have that
      \footnote{ By Schur's test  and self-adjointness of $\mathcal{T}$, we get (for $r>d$)
                     \begin{align*}
                        \|\mathcal{T}\|\leq \sup_{m\in{\Z}^d}\sum_{n\neq m}|m-n|^{-r}\leq \sum_{n\in{\Z}^d\setminus\{0\}}|n|^{-r}<\infty,
                      \end{align*}
                  where $\|\cdot\|$ is the standard operator norm on $\ell^2({\Z}^d)$.
                }
     $d<r<\infty$.
  \item  Let supp$(\mu)=\{x: \mu(x-\epsilon,x+\epsilon)>0 \text{ for any }\epsilon>0\}$ be the support of the common distribution $\mu$.  We assume that
  supp$(\mu)$ contains at least two points and supp $(\mu)$ is compact:
  \begin{equation*}
    \text{supp}(\mu)\subset[-M,M],\ \ \ \ 0<M<\infty.
  \end{equation*}
\end{itemize}

Under above assumptions, $\mathrm{H}_{\omega}$ is a bounded self-adjoint operator on $\ell^2(\mathbb{Z}^d)$ for each $\omega\in\Omega$. Denote by $\sigma(\mathrm{H}_{\omega})$ the spectrum of $\mathrm{H}_{\omega}$. A well-known result due to Pastur \cite{Pa80}  imply that there exists a set $\Sigma$ (compact and non-random) such that for $\mathbb{P}$ almost all $\omega$, $\sigma(\mathrm{H}_{\omega})=\Sigma$.

Let us recall the H\"older continuity of a distribution defined in \cite{CKM87}.

\begin{defn}[\cite{CKM87}]
We will say a probability measure $\mu$ is H\"older continuous of order $\rho>0$ if
\begin{align*}
\frac{1}{\mathcal{K}_\rho(\mu)}=\inf_{\kappa>0}\sup_{0<|a-b|\leq \kappa}|a-b|^{-\rho}\mu([a,b])<\infty.
\end{align*}
In this case will call $\mathcal{K}_\rho(\mu)>0$ the disorder of $\mu$.
\end{defn}
The main result of this paper is the following theorem.
\begin{thm}\label{mainthm}
Let ${\mathrm{H}}_{\omega}$ be defined by $(\ref{h})$ with the common distribution $\mu$ being H\"older continuous of order $\rho>0$, i.e., $\mathcal{K}_\rho(\mu)>0$. Assume $r\geq\max\{\frac{200d}{\rho}+25d,\ 1800d\}$. Fix any $0<\kappa<\mathcal{K}_\rho(\mu)$. Then there exists $\lambda_0=\lambda_0(\kappa,\mu,\rho,{M},r,d)>0$ such that for  $\lambda\geq\lambda_0$ and for $\mathbb{P}$  almost all  $\omega\in\Omega$, there exists a positive constant $q:=q(r)\leq \frac{r}{1600}$, such that for any $\phi\in\ell^2(\mathbb{Z}^d)$ satisfying $|\phi(n)|\leq C_{\phi}|n|^{-\theta}$, where $\theta\geq r/160$, there exists a constant $C_{\phi}=C_{\phi}(d,r,q,\theta)$ such that
\begin{equation}\label{dl}
\||X|^{q/2}e^{-\mathrm{i}\mathrm{H}_{\omega}t}\phi\|^2\leq C_{\phi},\ \quad\quad  \forall t\geq 0,
\end{equation}
where $X$ is the usual position operator.
\end{thm}
\begin{rem}
\begin{itemize}

\item[]
\item The core of Theorem \ref{mainthm} is to obtain a power-law version SULE condition, that is, all energy state of random operator \eqref{h} have the form of
\begin{equation*}
|\psi_{k,\omega}(n)|\leq D(\epsilon,\omega)|n_k|^{\epsilon}|n-n_k|^{-\alpha},
\end{equation*}
where $\epsilon>0$ and $D(\epsilon,\omega)$ is a finite constant that does not depend on energy state.

\item It should emphasized that the index $q$ in Theorem \ref{mainthm} is closely related to the index $r$ of operator $\mathcal{T}$.
This is different from the case where  exponential localization leads to dynamical localization, where there is no restriction on  the index $q$.
\end{itemize}
\end{rem}

\section{Preliminary Knowledge}

\subsection{Sobolev Norm of a Matrix}

Let $X_1, X_2\subset\mathbb{Z}^d$ be finite sets. Define
\begin{equation*}
  \textbf{M}^{X_1}_{X_2}=\{\mathcal{M}=(\mathcal{M}(k,k')\in\mathbb{C})_{k\in X_1,k'\in X_2}\}
\end{equation*}
to be the set of all complex matrices with row indexes in $X_1$ and column indexes in $X_2$. If $Y_1\subset X_1$, $Y_2\subset X_2$, we write $\mathcal{M}^{Y_1}_{Y_2}=(\mathcal{M}(k,k'))_{k\in Y_1, k'\in Y_2}$ for any $\mathcal{M}\in\mathbf{M}^{X_1}_{X_2}$.

\begin{defn}
Let $\mathcal{M}\in\mathbf{M}^{X_1}_{X_2}$. Define for $s\geq s_0$ the Sobolev norms of $\mathcal{M}$ as:
\begin{equation*}
  \|\mathcal{M}\|_{s}^2=C_0(s_0)\sum_{k\in X_1-X_2}\left(\sup_{k_1-k_2=k}|\mathcal{M}(k_1,k_2)|\right)^2\langle k\rangle^{2s},
\end{equation*}
where $\langle k\rangle=\max \{1,|k|\}$ and $C_0(s_0)>0$ is a constant depending on $s_0$.
\end{defn}

\subsection{Green's Function Estimate}

The Green's function plays a key role in spectral theroy. In this subsection we  present the first main result about  Green's function estimate.
For $n\in{\Z}^d$ and $L>0$, define the cube $\Lambda_L(n)=\{k\in{\Z}^d:\ |k-n|\leq L\}$. Moreover,   write $\Lambda_L=\Lambda_L(0)$. The volume of a finite set $\Lambda\subset{\Z}^d$
is defined to be $|\Lambda|=\# \Lambda$. We have $|\Lambda_L(n)|=(2L+1)^d$ ($L\in\N$) for example.

If $\Lambda\subset \Z^d$, denote ${\mathrm{H}}_{\Lambda}=R_{\Lambda}{\mathrm{H}_{\omega}}R_{\Lambda}$, where $R_{\Lambda}$ is the restriction operator. Define the Green's function (if it exists) as
\begin{align*}
G_{\Lambda}(E)=({\mathrm{H}}_{\Lambda}-E)^{-1},\  E\in\R.
\end{align*}

 Let us introduce  \textbf{good} cubes in ${\Z}^d$.
\begin{defn}
Fix $\tau'>0$, $\delta\in (0,1)$ and $d/2<s_0\leq r_1<r-d/2$. We call $\Lambda_L(n)$ is $(E,\delta)$-\textbf{good} if  $G_{\Lambda_L(n)}(E)$  exists and satisfies
\begin{align*}
\|G_{\Lambda_L(n)}(E)\|_s\leq L^{\tau'+\delta{s}} \ \mathrm{for}\ \forall\  s\in[s_0,r_1].
\end{align*}
Otherwise, we call $\Lambda_L(n)$ is $(E,\delta)$-\textbf{bad}.
We call $\Lambda_{L}(n)$ an $(E,\delta)$-\textbf{good} (resp.  $(E,\delta)$-\textbf{bad}) $L$-cube if it is $(E,\delta)$-\textbf{good} (resp. $(E,\delta)$-\textbf{bad}).
\end{defn}

\begin{rem}\label{rgdec}
Let $\zeta\in (\delta,1)$ and $\tau'-(\zeta-\delta)r_1<0$. Suppose that $\Lambda_L(n)$ is $(E,\delta)$-\textbf{good}. Then we have for $L\geq L_0(\zeta,\tau',\delta,r_1,d)>0$ and $|n'-n''|\geq L/2$,
\begin{align}\label{gdec}
|G_{\Lambda_L(n)}(E)(n',n'')|
&\leq |n'-n''|^{-(1-\zeta)r_1}.
\end{align}
\end{rem}

Assume the following relations hold true:
\begin{align}\label{para}
\left\{
\begin{aligned}
&-(1-\delta)r_1+\tau'+2s_0<0,\\
&-\xi r_1+\tau'+\alpha\tau+(3+\delta+4\xi)s_0<0,\\
&\alpha^{-1}(2\tau'+2\alpha\tau+(5+4\xi+2\delta)s_0)+s_0<\tau',\\
\end{aligned}
\right.
\end{align}
where $\alpha,\tau,\tau',r_1>1, \xi>0, s_0>d/2$ and $\delta\in(0,1)$.

\begin{defn}[]
We call a site $n\in\Lambda\subset{\Z}^d$ is $(l,E,\delta)$-\textbf{good} with respect to (\textit{w.r.t}) $\Lambda$ if there exists some $\Lambda_{l}(m)\subset\Lambda$ such that $\Lambda_{l}(m)$ is $(E,\delta)$-\textbf{good} and $n\in \Lambda_{l}(m)$ with ${\rm dist}(n,\Lambda\setminus\Lambda_{l}(m))\geq l/2$. Otherwise, we call $n\in\Lambda\subset{\Z}^d$ is $(l,E,\delta)$-\textbf{bad} \textit{w.r.t} $\Lambda$.
\end{defn}

Let
\begin{align}\label{tr}
\tau>(2p+(2+\rho)d)/\rho.
\end{align}
The multi-scale analysis argument on Green's function estimate is

\begin{thm}[Theorem 4.4 in \cite{Shi2020}]\label{msa}
Let $\mu$ be H\"older continuous of order $\rho>0$ (i.e., $\mathcal{K}_\rho(\mu)>0$). Fix $E_0\in\R$ with $|E_0|\leq 2(\|\mathcal{T}\|+M)$, and assume  \eqref{para}, \eqref{tr} hold true. Assume further that
\begin{equation}\label{pp}
(1+\xi)/\alpha\leq \delta,\quad p>\alpha d+2\alpha p/J
\end{equation}
for $J\in 2\N$. Then for $0<\kappa<\mathcal{K}_\rho(\mu)$, there exists $$\underline{L}_0=\underline{L}_0(\kappa,\mu,\rho,\|\mathcal{T}\|_{r_1},M,J,\alpha,\tau,\xi,\tau',\delta, p,{r_1},s_0,d)>0$$ such that the following holds: For $L_0\geq \underline{L}_0$,  there is some $\lambda_0=\lambda_0(L_0,\kappa,\rho,p,s_0,d)>0$ and $\eta=\eta(L_0,\kappa,\rho,p,d)>0$ so that for $\lambda\geq \lambda_0$ and $k\geq 0$, we have
\begin{align*}
 \mathbb{P}({\ \exists}\  E\in [E_0-\eta,E_0+\eta], \ {\rm s.t.\ }\  \Lambda_{L_k}(m)\  {\rm and}\  \Lambda_{L_k}(n)\ {\rm are\ }  (E,\delta)\text{-{\rm{bad}}})
  \leq L_k^{-2p}
  \end{align*}
for all $|m-n|>2L_k$, where $L_{k+1}=[L_k^{\alpha}]$ and $L_0\geq \underline{L}_0$.
\end{thm}

\subsection{Power-law localization}
Recall the Poisson's identity: Let $\psi=\{\psi(n)\}\in {\C}^{\Z^d}$ satisfy $\mathrm{H}_\omega \psi=E\psi$. Assume further $G_{\Lambda}(E)$ exists for some  $\Lambda\subset{\Z}^d$. Then for any $n\in\Lambda$, we have
\begin{align}\label{pi}
\psi(n)=-\sum_{n'\in\Lambda,n''\notin\Lambda}\lambda^{-1} G_\Lambda(E)(n,n')\mathcal{T}(n',n'')\psi(n'').
\end{align}
From Shnol's Theorem of \cite{Han19} in long-range operator case, to prove pure point spectrum of $\mathrm{H}_{\omega}$,  it suffices to show that each $\varepsilon$-generalized  eigenfunction belongs to $\ell^2({\Z}^d)$.
In \cite{Shi2020}, Shi shows that every $\varepsilon$-generalized (with $0<\varepsilon\leq c(d)\ll1$) eigenfunction $\psi$ of random operator \eqref{h} decays as  $|\psi(n)|\leq {|n|^{-{r}/{600}}}$ for $|n|\gg1$. Specifically, Shi obtains polynomially decaying of each generalized eigenfunction of $\mathrm{H}_{\omega}$ for $\mathbb{P}$ a.e. $\omega$. This yields the  \textit{power-law} localization:

\begin{thm}[Theorem 2.5 in \cite{Shi2020}]\label{pl}
Let ${\mathrm{H}}_{\omega}$ be defined by $(\ref{h})$ with the common distribution $\mu$ being H\"older continuous of order $\rho>0$, i.e., $\mathcal{K}_\rho(\mu)>0$. Let $r\geq\max\{\frac{100d+23\rho d}{\rho}, 331d\}$. Fix any $0<\kappa<\mathcal{K}_\rho(\mu)$. Then there exists $\lambda_0=\lambda_0(\kappa,\mu,\rho,{M},r,d)>0$ such that for  $\lambda\geq\lambda_0$, ${\mathrm{H}}_{\omega}$
has pure point spectrum  for $\mathbb{P}$ almost all  $\omega\in\Omega$. Moreover, for $\mathbb{P}$ almost all  $\omega\in\Omega$, there exists a complete system of eigenfunctions $\psi_{j,\omega}=\{\psi_{j,\omega}(n)\}_{n\in{\Z}^d},\ j=1,2,\cdots,$  satisfying
\begin{equation}\label{power-law}
|\psi_{j,\omega}(n)|\leq C_{j,\omega}|n|^{-r/600},\quad\quad |n|\gg 1.
\end{equation}
\end{thm}
\begin{rem}
Note that the coefficients $C_{j,\omega}$ in (\ref{power-law}) depend on the selection of energy state.
\end{rem}
\begin{lem}[Lemma A.1 in \cite{Shi2020}]
Let $L>2$ with $L\in\N$ and $\Theta-d>1$. Then we have that
\begin{align}\label{ldec}
\sum_{n\in{\Z}^d:\ |n|\geq L}|n|^{-\Theta}\leq C(\Theta,d)L^{-(\Theta-d)/2},
\end{align}
where $C(\Theta,d)>0$ depends only on $\Theta,d$.
\end{lem}

\section{Proof of Theorem \ref{mainthm}}

In order to obtain dynamical localization, it's need to control the location and the size of the boxes outside of which the eigenfunctions has ``effective" decrease.
\begin{thm}\label{upl}
For the operator $\mathrm{H}_{\omega}$ defined by (\ref{h}), we assume that $\kappa,\lambda_0$ and $\psi_{j,\omega},\ j\in\mathbb{N}$ satisfy Theorem \ref{pl}. Let $r\geq\max\{\frac{200d}{\rho}+25d, 331d\}$.
Then  for  $\lambda\geq\lambda_0$, there exists centers $n_{j,\omega}$ associated to the eigenfunctions $\psi_{j,\omega}$ with eigenvalues $E_{j,\omega}$ such that for any $\gamma\in[0,\frac{r}{160}]$ and any $\epsilon'\in[\frac{1}{3},\frac{1}{2})$, there exists  a constant   $C(\epsilon',\gamma)>0$ such that
\begin{equation}\label{upli}
  |\psi_{j,\omega}(n)|\leq C(\epsilon',\gamma)|n_{j,\omega}|^{\epsilon'\gamma}|n-n_{j,\omega}|^{-\gamma},\quad\quad \forall n\in\mathbb{Z}^d.
\end{equation}
Here, $C(\epsilon',\gamma)$ dose not depend on $j$(the eigenfunction).
\end{thm}

From Theorem \ref{upl}, the eigenfunctions $\psi_{j,\omega}$ are localized outside boxes of size $|n_{j,\omega}|/2$ around ``centers" $n_{j,\omega}$. This result is stronger than the power-law localization.

In the following, we choose appropriate parameters satisfying Remark \ref{rgdec}, \eqref{para}-\eqref{pp}. For this purpose, we can set by direct calculation that
\begin{align}\label{p}
\alpha=6,\quad \delta=\frac{1}{2},\quad \xi=2, \quad \zeta=\frac{19}{20}, \quad p=13d.
\end{align}
We define $J_\star=J_\star(d,\varepsilon)$ to be the smallest even integer satisfying  $p>6d+\frac{12}{J_\star}p$.
As a consequence, we can set
\begin{align}\label{tau}
\tau=\frac{29d}{\rho}+d,\quad s_0=\frac{3}{4}d, \quad \tau'=\frac{87d}{\rho}+7d,
\end{align}

\begin{equation}\label{r}
r\geq \max\left\{\frac{200d}{\rho}+25d,\ 331d\right\},\qquad r_1=r-8d.
\end{equation}

In order to prove Theorem \ref{mainthm}, we  need the following lemma which says that if $\psi$ is an eigenfunction of $\mathrm{H}_{\omega}$ with eigenvalue $E$, then $E$ must be close to the spectrum of $\mathrm{H}_{\Lambda_{L}(n)}$ provided $L$ is big enough and $\Lambda_{L}(n)$ is centered on a maximum of $|\psi(n)|$.

\begin{lem}\label{bad}
There exists a constant $L_{\star}(d,r)$ so that if $\psi\in\ell^{2}(\mathbb{Z}^d)$ is an eigenfunction of $\mathrm{H}_{\omega}$ with eigenvalue $E$, and $n_{\star}$ satisfies $|\psi(n_{\star})|=\sup\{|\psi(n)|,\ n\in\mathbb{Z}^d\}$, then $\Lambda_{L}(n_{\star})$ is $(E,1/2)$-{bad} for all $L\geq L_{\star}(d,r)$.
\end{lem}

\begin{proof}
Let $\psi\in\ell^2(\mathbb{Z}^d)$ be as in the lemma, so $n_{\star}$ exists. Suppose that $\Lambda_{L_k}(n_{\star})$ is $(E,1/2)$-good for some  $k\geq k_1=k_1(d,r)$ sufficiently large. Applying the identity (\ref{pi}) at the point $n_{\star}$, one has that
\begin{align*}
|\psi(n_{\star})|&\leq \sum_{n'\in \Lambda_{L_k}(n_{\star}),\atop n''\notin \Lambda_{L_k}(n_{\star})}C(d)  |G_{\Lambda_{L_k}(n_{\star})}(E)(n_{\star},n')|\cdot|n'-n''|^{-r}|\psi(n'')|\\
&\leq\ {\rm (V)}+{\rm (VI)},
\end{align*}
where
\begin{align*}
{\rm (V)}&=\sum_{|n'-n_{\star}|\leq L_k/2, \atop|n''-n_{\star}|>L_k} C(d,s_0)L_k^{\tau'+\frac{1}{2} s_0}|n'-n''|^{-r}|\psi(n_{\star})|,\\
{\rm (VI)}&=\sum_{L_k/2< |n'-n_{\star}|\leq L_k , \atop |n''-n_{\star}|>{L_k}} C(d) |n'-n_{\star}|^{-\frac{r_1}{20}}|n'-n''|^{-r}|\psi(n_{\star})|.
\end{align*}
When $|n'-n_{\star}|\leq L_k/2$ and $|n''-n_{\star}|>L_k$, one has that $|n'-n''|\geq |n''-n_{\star}|-|n'-n_{\star}|>|n''-n_{\star}|/2$.
By \eqref{ldec}, we have that
\begin{align*}
{\rm (V)}&\leq\sum_{|n''-n_{\star}|>L_k} C(r,d,s_0)L_k^{d+\tau'+\frac{1}{2} s_0}|n''-n_{\star}|^{-r}|\psi(n_{\star})|\\
&\leq C(d,r,s_0)L_k^{-\frac{r}{2}+\tau'+\frac{1}{2} s_0+\frac{3}{2}d}|\psi(n_{\star})|,
\end{align*}
For the term (VI), one has that
\begin{align*}
{\rm (VI)}=&\ (\sum_{L_k/2<|n'-n_{\star}|\leq L_k , \atop L_k<|n''-n_{\star}|<2L_k}+\sum_{L_k/2<|n'-n_{\star}|\leq L_k , \atop |n''-n_{\star}|\geq{2L_k}}) C(d)|n'-n_{\star}|^{-\frac{r_1}{20}}|n'-n''|^{-r}|\psi(n_{\star})|\\
\leq &\ \sum_{L_k/2< |n'-n_{\star}|\leq L_k }C(d,r)L_k^d|n'-n_{\star}|^{-\frac{r_1}{20}}|\psi(n_{\star})|\\
&\ \ +\sum_{L_k/2<|n'-n_{\star}|\leq L_k , \atop |n''-n_{\star}|\geq{2L_k}}C(d,r)  |n'-n_{\star}|^{-\frac{r_1}{20}}|n''-n_{\star}|^{-r}|\psi(n_{\star})|\\
\leq &\ C(d,r_1)L_k^{-\frac{r_1}{20}+2d}|\psi(n_{\star})|+C(d,r_1,r)L_k^{-\frac{r_1}{20}-\frac{r}{2}+\frac{3}{2}d}|\psi(n_{\star})|\\
\leq&\ C(d,r_1,r)L_k^{-\frac{r_1}{20}+2d}|\psi(n_{\star})|.
\end{align*}
Recalling (\ref{p})-(\ref{r}), one has that
\begin{equation*}
-\frac{r}{2}+\tau'+\frac{1}{2} s_0+\frac{3}{2}d<-\frac{r}{25},\quad -\frac{r_1}{20}+2d<-\frac{r}{25}.
\end{equation*}
Hence, for  large enough $k$ (depending on $d$ and $r$),
$$|\psi(n_{\star})|\leq L_k^{-\frac{r}{30}}|\psi(n_{\star})|<|\psi(n_{\star})|,$$
which is impossible. Therefore $\Lambda_{L_{k}}(n_{\star})$ is $(E,1/2)$-bad for all $k\geq k_1(d,r)$. The lemma is proved.
\end{proof}

Now we can give the proof of Theorem \ref{upl}.

\begin{proof}[{\bf Proof of Theorem \ref{upl}}]
Under the hypotheses of Theorem \ref{upl} and Theorem \ref{pl}, $\mathrm{H}_{\omega}$ has  power-law localization for $\mathbb{P}$ almost all $\omega\in\Omega$. This means that there exists $\Omega_1\subset\Omega$, $\mu(\Omega_1)=1$ so that for all $\omega\in\Omega_1$, $\sigma_c(\mathrm{H}_{\omega})=\emptyset$ and for all eigenvalues $E_{j,\omega}$, the corresponding eigenfunction $\psi_{j,\omega}$ is $\ell^2$ and satisfies (\ref{power-law}) with $\|\psi_{j,\omega}\|=1$.

Our ultimate goal is  to control the constant $C_{j,\omega}$ in (\ref{power-law}). More precisely, for each $0<\gamma\leq r/400$, we try to  show that $n_{j,\omega}$ can be chosen so that $C_{j,\omega}$ grows slower than $|n_{j,\omega}|^{\epsilon'\gamma}$ for  $1/3\leq\epsilon'<1/2$.

The outline of the proof is  a little complicated, for the convenience of reader, we will divide the main proof into three parts. Firstly, we fix $L_0=\underline{L}_0$, $\lambda_0,\ \eta$ and $I=[E_0-\eta,E_0+\eta]$ in Theorem \ref{msa}. Recalling Theorem \ref{msa}, we have for $\lambda\geq\lambda_0$ and $k\geq 0$,
\begin{align*}
 \mathbb{P}({\ \exists}\  E\in I, \ {\rm s.t.\ both }\ \Lambda_{L_k}(m)\  {\rm and}\  \Lambda_{L_k}(n)\ {\rm are\ }  (E,\delta)\text{-{\rm{bad}}})
  \leq L_k^{-2p}
  \end{align*}
for all $|m-n|>2L_k$, where $L_{k+1}=[L_k^{\alpha}]$ and $L_0\gg 1$.

\textbf{Step one}.
For any $k\geq 0$, we define the set
\begin{equation*}
A_{k+1}(n_0)=\Lambda_{2L_{k+1}}(n_0)\setminus\Lambda_{L_k}(n_0)
\end{equation*}
and the event:
\begin{equation}\label{ek}
{\mathbf{E}}_k(n_0)=\{ \exists E, \ \exists n\in A_{k+1}(n_0), \text{ s.t.\ both\  }  \Lambda_{L_k}(n_0)\text{ and }  \Lambda_{L_k}(n)\text{ are }(E,1/2)\text{-{bad}}\}.
\end{equation}
From Theorem \ref{msa},
\begin{equation*}
\mathbb{P}({\mathbf{E}}_k(n_0))\leq \sum_{n\in A_{k+1}(n_0)}L_k^{-2p}\leq C(d)(4L_k^\alpha+1)^dL_k^{-2p}\leq C(d)L_k^{-2p+\alpha d}.
\end{equation*}
For  $1/3\leq \epsilon'<1/2$, we define
$$F_{k}=\bigcup_{|n_0|\leq L_{k+1}^{1/\epsilon'}}\mathbf{E}_{k}(n_0).$$
Then  $$\mathbb{P}(F_{k})\leq \sum_{|n_0|\leq L_{k+1}^{1/\epsilon'}}\mathbb{P}(E_{k}(n_0))\leq C(d,\epsilon')L_{k}^{-2p+\alpha d(1+\frac{1}{\epsilon'})}, $$
where  $p$ and $\alpha$ are defined in (\ref{p}). It is easy to verified that  $\sum_{k=0}^{\infty}\mathbb{P}(F_k)< \infty$. Then, the Borel-Cantelli lemma  implies that
\begin{equation*}
  \mathbb{P}\left(\lim_{m\rightarrow\infty}\bigcup_{k\geq m}F_k\right)=0,
\end{equation*}
so that the set
\begin{equation*}
  \Omega_2=\{\omega\in\Omega_1:\ \exists \tilde{k}_1=\tilde{k}_1(p,d,\epsilon'), \text{ s.t. } \forall k\geq \tilde{k}_1,\ \omega\notin F_k\}
\end{equation*}
has full measure.

Now pick $\omega\in\Omega_1\cap\Omega_2$, which will be kept fixed throughout the rest of the proof. Let $\psi_{j,\omega}$ be the  eigenfunction of energy $E_{j,\omega}$, and $n_{j,\omega}$ be a point where $|\psi_{j,\omega}(n_{j,\omega})|$ is maximal. Note that such a point exists, since $\omega\in\Omega_1$ and $\psi_{j,\omega}\in\ell^2(\mathbb{Z}^d)$. Furthermore, we define the integers
\begin{equation}\label{k0}
  \hat{k}_2(\epsilon',m)=\min\{k \geq 0 \text{ such that } |m|^{\epsilon'}<L_{k+1}\},\quad m\in\mathbb{Z}^d,
\end{equation}
and
\begin{equation}\label{k2}
 \bar{ k}_2=\bar{k}_2(p,d,\epsilon',n_{j,\omega})=\max\{\tilde{k}_1, \hat{k}_2(\epsilon',n_{j,\omega})\}.
\end{equation}
For any $k\geq \bar{k}_2$,  we see that $\omega\notin E_{k}(n_{j,\omega})$ from the definition of $F_k$. By (\ref{ek}), $\forall k\geq \bar{k}_2$  and $\forall n\in A_{k+1}(n_{j,\omega})$, either $\Lambda_{L_k}(n_{j,\omega})$  or $\Lambda_{L_k}(n)$ is $(E_{j,\omega},1/2)$-good. Applying Lemma \ref{bad}, there is an integer $$k_2=\max\{k_1,\bar{k}_2\}=\max\{k_1,\tilde{k}_2,\hat{k}_2\},$$
where $k_1={k}_1(d,r)$ and $\tilde{k}_2:=\tilde{k}_2(p,d,r,\epsilon')$ does not depend on $j$, such that for any energy $E_{j,\omega}$,
\begin{equation*}
  \Lambda_{L_k}(n) \text{ is }(E_{j,\omega},1/2)\text{-good}, \quad \forall n\in A_{k+1}(n_{j,\omega}),\quad \forall k\geq k_2.
\end{equation*}

\textbf{Step two}.  Let us apply the Possion's identity (\ref{pi}) at the point $n\in A_{k+1}(n_{j,\omega})$. Similarly to the proof of Lemma \ref{bad}, and recalling (\ref{p})-(\ref{r}), for any $k\geq k_2$, one has that
\begin{align*}
|\psi_{j,\omega}(n)|&\leq \sum_{n'\in \Lambda_{L_k(n)}, \atop n''\notin \Lambda_{L_k(n)}}  |G_{\Lambda_{L_k}(n)}(E_{j,\omega})(n,n')|\cdot|n'-n''|^{-r}|\psi_{j,\omega}(n'')|\\
&\leq C(r,d,s_0)L_k^{-\frac{r}{2}+\tau'+\frac{1}{2} s_0+\frac{3d}{2}}+ C(d,r,r_1)L_k^{-\frac{r_1}{20}+2d}\\
&\leq C(d,r)L_k^{-\frac{r}{25}}.
\end{align*}
 Set
 \begin{equation}
 \widetilde A_{k+1}(n)=\Lambda_{\frac{8}{5}L_{k+1}}(n) \setminus \Lambda_{\frac{4}{3}L_k}(n)\subset A_{k+1}(n).
 \end{equation}
If $n\in \widetilde{A}_{k+1}(n_{j,\omega})$, one has that $L_k\geq (\frac{5}{8}|n-n_{j,\omega}|)^{\frac{1}{6}}$, and
\begin{equation*}
  |\psi_{j,\omega}(n)|\leq C(d,r) |n-n_{j,\omega}|^{-\frac{r}{150}}.
\end{equation*}
Hence, one can find $k_3=\max\{\tilde{k}_3,k_2\}$, where  $\tilde{k}_3=\tilde{k}_3(d,r)$ is independent of $j$, such that
\begin{equation*}
  |\psi_{j,\omega}(n)|\leq  |n-n_{j,\omega}|^{-\frac{r}{160}},\qquad \forall k\geq {k}_3.
\end{equation*}
Then for any $0<\gamma\leq r/160$,
\begin{equation*}
  |\psi_{j,\omega}(n)|\leq |n-n_{j,\omega}|^{-\gamma},\quad \forall n\in\widetilde{A}_{k+1}(n_{j,\omega}),\quad \forall k\geq k_3.
\end{equation*}
Since $\bigcup\limits_{k\geq k_3}\widetilde{A}_{k+1}(n_{j,\omega})=\{n\in\mathbb{Z}^d:\ |n-n_{j,\omega}|>\frac{4}{3}L_{k_3}\}$, there exists $k_4=\max\{\tilde{k}_4,k_3\}$, where $\tilde{k}_4$ does not depend on $j$, such that
\begin{equation*}
  |\psi_{j,\omega}(n)|\leq |n-n_{j,\omega}|^{-\gamma},\qquad \forall n: |n-n_{j,\omega}|\geq L_{k_4}.
\end{equation*}
\textbf{Step three}. Using the fact  that $|\psi_{j,\omega}(n)|\leq 1$ for all $n\in\mathbb{Z}^d$, one has
\begin{equation}\label{jn}
  |\psi_{j,\omega}(n)|\leq C(\epsilon',\gamma)  L_{k_4}^{\gamma} |n-n_{j,\omega}|^{-\gamma}, \qquad \forall  n\in\mathbb{Z}^d.
\end{equation}
Now, we try to control the $j$-dependence of the constant $L_{k_4}^{\gamma}$. Note that the only $j$-dependence of $k_4$ comes from $\hat{k}_2(\epsilon',n_{j,\omega})$. Suppose $\sup\{|n_{j,\omega}| \}<\infty$, then $k_4$ can be chosen $j$-independently, so that we actually obtain a uniform localization for the all energy state
\begin{equation*}
  |\psi_{j,\omega}(n)|\leq C(\epsilon',\gamma) |n-n_{j,\omega}|^{-\gamma}, \quad \forall n\in\mathbb{Z}^d.
\end{equation*}
But the following Lemma \ref{njomega} contradicts this first possibility. So, in fact, $\sup\{|n_{j,\omega}|\}=\infty$, and for $j$ sufficiently large, one has
\begin{equation*}
 k_4=\hat{k}_{2}(\epsilon',n_{j,\omega}),
\end{equation*}
and recalling the definition of $\hat{k}_{2}(\epsilon',n_{j,\omega})$ in (\ref{k0}),
\begin{equation*}
  L_{k_4}\leq |n_{j,\omega}|^{\epsilon'}.
\end{equation*}
Inserting this in (\ref{jn}) yields the announced result. Theorem \ref{upl} is proved.
\end{proof}

We also need to control  the growth of $|n_{j,\omega}|$ with $j$, which is given by the following preliminary lemma:

\begin{lem}\label{njomega}
 Assume that  $n_{j,\omega}$ are defined in Theorem \ref{upl}. Then one can order $|n_{j,\omega}|$ in increasing order such that for $j$ larger enough,
 \begin{equation*}
   |n_{j,\omega}|\geq C |j|^{\frac{1}{d}}.
 \end{equation*}
\end{lem}
\begin{proof}
From Theorem \ref{upl}, $\{\psi_{j,\omega}:j=1,2,\cdots\}$ is a  complete normalized orthogonal system of $\ell^2(\mathbb{Z}^d)$ and each $n_{j,\omega}$ is chosen so that $|\psi_{j,\omega}(n_{j,\omega})|=\sup\{|\psi_{j,\omega}(n)|,\ n\in\mathbb{Z}^d\}$.
Therefore we have
\begin{equation}\label{con1}
\ \ \ \ \ \sum_{n\in\mathbb{Z}^d}|\psi_{j,\omega}(n)|^2=1,\qquad \forall j=1,2,\cdots,
\end{equation}
\begin{equation}\label{con2}
 \sum_{j=1}^{\infty}|\psi_{j,\omega}(n)|^2=1,\qquad \forall n\in\mathbb{Z}^d.
\end{equation}

Fix $0<\varepsilon<1$. For $j\in\mathbb{N}\setminus \{0\}$, by using Theorem \ref{upl} and (\ref{ldec}), one has that
\begin{align*}
  \sum_{|n-n_{j,\omega}|\geq \varepsilon L}|\psi_{j,\omega}(n)|^2&\leq C(\epsilon',\gamma)\sum_{|n-n_{j,\omega}|\geq \varepsilon L}|n_{j,\omega}|^{2\epsilon'\gamma}|n-n_{j,\omega}|^{-2\gamma}\\
  &\leq C(d,\epsilon',\gamma)\varepsilon^{-\gamma+\frac{d}{2}}|n_{j,\omega}|^{2\epsilon'\gamma}L^{-\gamma+\frac{d}{2}}.
\end{align*}
Assume $|n_{j,\omega}|\leq L$. Then
\begin{equation*}
  \sum_{|n|\geq (1+\varepsilon)L}|\psi_{j,\omega}(n)|^2\leq \sum_{|n-n_{j,\omega}|\geq \varepsilon L}|\psi_{j,\omega}(n)|^2
  \leq C(d,\epsilon',\gamma)\varepsilon^{-\gamma+\frac{d}{2}}L^{-\gamma+2\epsilon'\gamma+\frac{d}{2}}.
\end{equation*}
From (\ref{con2}),
\begin{equation*}\begin{aligned}
  (2(1+\varepsilon)L+1)^d & =\sum_{j\in\mathbb{N}\setminus\{0\}\atop |n|\leq (1+\varepsilon)L}|\psi_{j,\omega}(n)|^2 \geq \sum_{j:|n_{j,\omega}|\leq L \atop |n|\leq (1+\varepsilon)L}|\psi_{j,\omega}(n)|^2\\
   & \geq \#\{j:|n_{j,\omega}|\leq L\}\min_{j:|n_{j,\omega}|\leq L}  \sum_{|n|\leq (1+\varepsilon)L}|\psi_{j,\omega}(n)|^2\\
   & \geq \#\{j:|n_{j,\omega}|\leq L\}  \left(1-\max_{j:|n_{j,\omega}|\leq L}\sum_{|n|\geq (1+\varepsilon)L}|\psi_{j,\omega}(n)|^2\right)\\
   & \geq \#\{j:|n_{j,\omega}|\leq L\} (1-C(d,\epsilon',\gamma)\varepsilon^{-\gamma+\frac{d}{2}}L^{-\gamma+2\epsilon'\gamma+\frac{d}{2}}).
\end{aligned}\end{equation*}
Choose $\varepsilon=1/2$, $\epsilon'=1/3$ and $\gamma=r/160$. Since $r\geq 331d$, we have $-\gamma+2\epsilon'\gamma+\frac{d}{2}\leq -\frac{3}{16}d$. Then there exists $L_0=L_0(d,r)$ large enough, such that
\begin{equation}\label{nj}
  \#\{j:|n_{j,\omega}|\leq L\}\leq C(d)L^d,  \quad \forall L\geq L_0,
\end{equation}
where $C(d)$ is independent of $L$ and $j$. This tells us that $L\geq L_0$, $N(L)=\#\{j:|n_{j,\omega}|\leq L\}$ is finite and we can reorder the eigenfunctions so $|n_{j,\omega}|$ is increasing. Therefore, one has
\begin{equation*}
  |n_{j,\omega}|\geq c(d)j^{\frac{1}{d}},
\end{equation*}
for $j$ large enough (depending on $d$ and $r$).
\end{proof}

Finally, we can give the complete proof of Theorem \ref{mainthm}.
\begin{proof}[\emph{\textbf{Proof of Theorem \ref{mainthm}}}]
Let $\phi\in\ell^2(\mathbb{Z}^d)$ be such that, for some constant $C_{\phi}>0$ and $\theta\geq r/160 $, $|\phi(n)|\leq C_{\phi}|n|^{-\theta}$. We have to bound $\|X^{q/2}e^{-\mathrm{i}\mathrm{H}_{\omega}t}\phi\|$, for some $q>0$ and all $t>0$. Since $e^{-\mathrm{i}\mathrm{H}_{\omega}t}\phi=\sum_{j}e^{-iE_{j,\omega}t}\langle\phi,\ \psi_{j,\omega}\rangle \psi_{j,\omega}$ and $\|e^{-i\mathrm{H}_{\omega}t}\phi\|= \|\phi\|$, one has that
\begin{equation*}\begin{aligned}
  \|X^{q/2}e^{-\mathrm{i}\mathrm{H}_{\omega}t}\phi\|^2 & = \langle X^{q} e^{-\mathrm{iH}_{\omega}t}\phi, e^{-\mathrm{iH}_{\omega}t}\phi\rangle \\
  & \leq \sum_{j}|\langle\phi, \psi_{j,\omega}\rangle| \  |\langle X^q\psi_{j,\omega},\ e^{-\mathrm{iH}_{\omega}t}\phi \rangle|\\
  &\leq \|\phi\|\sum_{j} |\langle\phi,\ \psi_{j,\omega}\rangle|\|X^q \psi_{j,\omega}\|.
\end{aligned}\end{equation*}
From Theorem \ref{upl}, we can choose $\epsilon'=1/3$. Then there exists  a constant   $C(\gamma)>0$ such that
\begin{equation}\label{upli'}
  |\psi_{j,\omega}(n)|\leq C(\gamma)|n_{j,\omega}|^{\frac{\gamma}{3}}|n-n_{j,\omega}|^{-\gamma},\quad\quad \forall n\in\mathbb{Z}^d.
\end{equation}
One has that
\begin{equation*}
\|X^q \psi_{j,\omega}\|^2  =\sum_{n\in\mathbb{Z}^d}|n^q\psi_{j,\omega}(n)|^2 \leq C(\gamma)|n_{j,\omega}|^{\frac{2\gamma}{3}}\sum_{n\in\mathbb{Z}^d} |n|^{2q}  |n-n_{j,\omega}|^{-2\gamma},
\end{equation*}
where
\begin{equation*}\begin{aligned}
  \sum_{n\in\mathbb{Z}^d} |n|^{2q}  |n-n_{j,\omega}|^{-2\gamma}&\leq \sum_{|n-n_{j,\omega}|\geq 2|n_{j,\omega}|} |n|^{2q}  |n-n_{j,\omega}|^{-2\gamma}
       +\sum_{|n-n_{j,\omega}|< 2| n_{j,\omega}|} |n|^{2q}\\
       &\leq C(d,q,\gamma)|n_{j,\omega}|^{-\gamma+q+\frac{d}{2}} + C(d,q)|n_{j,\omega}|^{2q+d}\\
       &\leq C(d,q,\gamma)|n_{j,\omega}|^{2q+d}.
\end{aligned}\end{equation*}
Therefore
\begin{equation*}
 \|X^q \psi_{j,\omega}\|\leq C(d,r,q,\gamma) |n_{j,\omega}|^{q+\frac{\gamma}{3}+\frac{d}{2}}.
\end{equation*}
Moreover, from the assumption of $\phi$, we have
\begin{equation*}
  |\langle\phi,\psi_{j,\omega}\rangle|\leq \sum_{n\in\mathbb{Z}^d}|\phi(n)||\psi_{j,\omega}(n)|
                \leq C_{\phi}(\gamma)|n_{j,\omega}|^{\frac{\gamma}{3}}\sum_{n}|n|^{-\theta}|n-n_{j,\omega}|^{-\gamma}
\end{equation*}
where
\begin{equation*}\begin{aligned}
  \sum_{n}|n|^{-\theta}|n-n_{j,\omega}|^{-\gamma}& \leq \sum_{|n-n_{j,\omega}|<|n_{j,\omega}|/2} |n|^{-\theta}
         +\sum_{|n_{j,\omega}|\leq |n-n_{j,\omega}|\leq  2|n_{j,\omega}|} |n-n_{j,\omega}|^{-\gamma}\\
   &\ \ \ +\sum_{|n-n_{j,\omega}|\geq  2|n_{j,\omega}|}  |n|^{-\theta}|n-n_{j,\omega}|^{-\gamma}\\
   &\leq C(d,\theta)|n_{j,\omega}|^{-\theta+d}+ C(d,\gamma)|n_{j,\omega}|^{-\gamma+d}+
            C(d,\theta,\gamma)|n_{j,\omega}|^{-\frac{\theta}{2}-\frac{\gamma}{2}+\frac{d}{2}}\\
   &\leq C(d,\theta,\gamma) |n_{j,\omega}|^{-\gamma+d}.
\end{aligned}\end{equation*}
Therefore
\begin{equation*}
   |\langle\phi,\psi_{j,\omega}\rangle|\leq C_{\phi}(d,\theta,\gamma)|n_{j,\omega}|^{-\frac{2\gamma}{3}+d}.
\end{equation*}
Choose $0<q\leq\gamma/10$ and  $\gamma=r/160$. Since  $r\geq 1800d$, we have $-\frac{\gamma}{3}+q+\frac{3d}{2}<-\frac{11d}{10}$. Recalling Lemma \ref{njomega}, one has that
\begin{align*}
 \sum_{j} |\langle\phi,\ \psi_{j,\omega}\rangle|\|X^q \psi_{j,\omega}\| & \leq C_{\phi}(d,r,q,\theta)\sum_{j} |n_{j,\omega}|^{-\frac{\gamma}{3}+q+\frac{3d}{2}}\\
 &\leq C_{\phi}(d,r,q,\theta)\sum_{j} |n_{j,\omega}|^{-\frac{11d}{10}}\\
 &\leq C_{\phi}(d,r,q,\theta)\sum_{j}|j|^{-\frac{11}{10}}\leq C_{\phi}(d,r,q,\theta).
\end{align*}
Therefore
$$\|X^{q/2}e^{-\mathrm{iH}_{\omega}t}\phi\|^2\leq C_{\phi}(d,r,q,\theta),\quad \forall t\geq 0.$$
The proof of  Theorem \ref{mainthm} is finished.
\end{proof}

\section*{Acknowledgements}
This work was supported by China Postdoctoral Science Foundation (No.2021M692717).

\bibliographystyle{abbrv} 

\end{document}